\documentclass[10pt,aps,pra,twocolumn,floatfix,showpacs,groupedaddress,amsmath,amssymb]{revtex4-1}

\usepackage{amsthm}

\usepackage{graphicx}

\newtheorem{theorem}{Theorem}

\newtheorem{proposition}[theorem]{Proposition}

\newenvironment{definition}[1][Definition.]{\begin{trivlist}
\item[\hskip \labelsep {\bfseries #1}]}{\end{trivlist}}
\newenvironment{example}[1][Example.]{\begin{trivlist}
\item[\hskip \labelsep {\bfseries #1}]}{\end{trivlist}}
\newenvironment{remark}[1][Remark.]{\begin{trivlist}
\item[\hskip \labelsep {\bfseries #1}]}{\end{trivlist}}

\DeclareMathOperator{\tr}{Tr}
\DeclareMathOperator{\identity}{I}

\newcommand{\ket} [1] {| #1 \rangle}
\newcommand{\bra} [1] {\langle #1 |}
\newcommand{\braket}[2]{\langle #1 | #2 \rangle}

\newcommand{\proj}[1]{\mbox{$|#1\rangle \!\langle #1 |$}}
\makeatletter

\begin{document}


\title{Entanglement negativity and topological order}


\author{Yirun Arthur Lee}
\affiliation{Perimeter Institute for Theoretical Physics, Waterloo, Ontario, N2L 2Y5, Canada}

\author{Guifre Vidal}
\affiliation{Perimeter Institute for Theoretical Physics, Waterloo, Ontario, N2L 2Y5, Canada}


\date{\today}

\begin{abstract}
We use the entanglement negativity, a measure of entanglement for mixed states, to probe the structure of entanglement in the ground state of a topologically ordered system. Through analytical calculations of the negativity in the ground state(s) of the toric code model, we explicitly show that the pure-state entanglement of a region $A$ and its complement $B$ is the sum of two types of contributions: \textit{boundary} entanglement and \textit{long-range} entanglement. \textit{Boundary entanglement} is seen to be insensitive to tracing out the degrees of freedom in the interior of regions $A$ and $B$, and therefore it only entangles degrees of freedom in $A$ and $B$ that are close to their common boundary. We recover the well-known result that boundary entanglement is proportional to the size of each boundary separating $A$ and $B$ and it includes an additive, universal correction. The second, long-range contribution to pure-state entanglement appears only when $A$ and $B$ are non-contractible regions (e.g. on a torus) and it is seen to be destroyed when tracing out a non-contractible region in the interior of $A$ or $B$. In the toric code, only the long-range contribution to the entanglement depends on the specific ground state under consideration.
\end{abstract}


\maketitle

\section{Introduction}


The study of entanglement in quantum many-body systems has in recent years become a highly interdisciplinary endeavor. By studying the scaling of ground state entanglement, information about the universality class of both quantum phases transitions \cite{C1,C2,C3} and topologically ordered phases of matter \cite{Hamma, TEE1, TEE2} can be obtained (see \cite{Reviews} for reviews). Moreover, insights into the structure of entanglement has led to new ways of describing and numerically simulating many-body states \cite{TensorNetworks}.



Much of our present understanding of many-body entanglement is based on studying the entanglement between a region $A$ of a system and its complement $B$. The state $\ket{\Psi}$ of a many-body system can be canonically written in its \textit{Schmidt decomposition}
\begin{equation}\label{eq:Schmidt}
    \ket{\Psi} = \sum_{\alpha} \sqrt{p_{\alpha}} \ket{\phi_{\alpha}}_A\otimes \ket{\varphi_{\alpha}}_B,
\end{equation}
where $\{\ket{\phi_{\alpha}}_A\}$ and $\{\ket{\varphi_{\alpha}}_B\}$ are sets of orthonormal states in $A$ and $B$,  $\braket{\phi_{\alpha}}{\phi_{\alpha'}} = \braket{\varphi_{\alpha}}{\varphi_{\alpha'}} = \delta_{\alpha,\alpha'}$, and $\sqrt{p_{\alpha}}$ are the Schmidt coefficients, with $p_\alpha\geq 0$, $\sum_{\alpha} p_{\alpha} = 1$. It follows from Eq. \ref{eq:Schmidt} that the reduced density matrices $\rho_A \equiv \tr_B \proj{\Psi}$ and $\rho_B \equiv \tr_A \proj{\Psi}$ for regions $A$ and $B$ have the same eigenvalue spectrum,
\begin{eqnarray}\label{eq:reduced}
    \rho_A = \sum_{\alpha} p_{\alpha} \ket{\phi_{\alpha}}\bra{\phi_{\alpha}},  \\
    \rho_B = \sum_{\alpha} p_{\alpha} \ket{\varphi_{\alpha}}\bra{\varphi_{\alpha}}.
\end{eqnarray}
We can then use the von Neumann entropy of $\rho_A$,
\begin{equation}\label{eq:vonNeumann}
S(\rho_A) \equiv -\tr\left( \rho_A \log_2 (\rho_A) \right) = - \sum_\alpha p_{\alpha} \log_2 (p_{\alpha}),
\end{equation}
and, more generally, the Renyi entropy of order $q$,
\begin{equation}\label{eq:Renyi}
S_q(\rho_A) \equiv \frac{1}{1-q}  \log_2 \tr \left( (\rho_A)^{q}  \right) = \frac{1}{1-q} \log_2 \left( \sum_\alpha (p_{\alpha})^q \right),
\end{equation}
to quantify the amount of entanglement between $A$ and $B$.

Consider now a many-body system divided into three regions: regions $A$ and $B$, and the rest of the system, $C$. Assume that $A\cup B \cup C$ is in a pure state $\ket{\Psi}$, and let $\rho_{AB} \equiv \tr_C \proj{\Psi}$ be the state of $A\cup B$. If part $A\cup B$ is entangled with $C$, then $\rho_{AB}$ is a mixed state. We would like to quantify the entanglement between $A$ and $B$ contained in $\rho_{AB}$. However, we can no longer use the entropy of the reduced density matrix $\rho_A = \tr_B (\rho_{AB})$ (or $\rho_B \equiv \tr_A (\rho_{AB})$) to do so, since this entropy quantifies the entanglement between $A$ and $B\cup C$ (respectively, between $B$ and $A\cup C$). Although it is still possible to use entropy-based measures, such as the mutual information $S(\rho_A) + S(\rho_B) - S(\rho_{AB})$, to characterize the total amount of correlations between $A$ and $B$, these measures cannot distinguish between quantum entanglement and classical correlations.


Given the mixed state $\rho_{AB}$, with components
\begin{equation}\label{eq:rhoAB}
    \rho_{AB} = \sum_{ijkl} (\rho_{AB})_{ijkl} \ket{i_A \otimes j_B}\bra{k_A \otimes l_B},
\end{equation}
its partial transposition $\rho^{T_A}_{AB}$ is defined to have coefficients
\begin{equation}\label{eq:transpose}
    (\rho_{AB}^{T_A})_{kjil} \equiv (\rho_{AB}^{T_A})_{ijkl}.
\end{equation}
A sufficient condition for $\rho_{AB}$ to be entangled is that its partial transposition has at least one negative eigenvalue $n<0$ \cite{Peres}, that is,
\begin{equation}\label{eq:Peres}
    \rho_{AB}^{T_A} \not\geq 0.
\end{equation}
Based on this observation, we could use the sum of negative eigenvalues $\{n_i\}$ of $\rho_{AB}^{T_A}$, called \textit{negativity} $\mathcal{N}(\rho_{AB})$ of $\rho_{AB}$,
\begin{equation}\label{eq:negativity}
    \mathcal{N}(\rho_{AB}) \equiv \sum_{i} |n_i|,
\end{equation}
to characterize the entanglement between regions $A$ and $B$. The negativity, first introduced in \cite{Zyczkowski}, is of interest as a measure of mixed-state entanglement because it can only decrease under local manipulations of subsystems $A$ and $B$ \cite{VidalWerner,Lee,Eisert,Plenio}, as shown by several authors in the context of quantum information \cite{History}. As described in \cite{VidalWerner}, an equivalent quantity, the logarithmic negativity
\begin{equation}\label{eq:log-negativity}
    E_{\mathcal{N}} \equiv \log_2 (1+2\mathcal{N}),
\end{equation}
is an upper bound to how much pure-state entanglement can be distilled from a mixed state, and therefore it has an operational meaning.

Recently, using the replica trick, Ref. \cite{CFT1} presented analytical calculations of the negativity of two intervals in 1+1 quantum field theories (see also \cite{CFT2, MonteCarlo} for related numerical computations). These calculations are important because they show that the negativity can be used to extract universal properties of quantum critical systems, possibly beyond what has been possible through entropy calculations.

The goal of this paper is to use the negativity to investigate the structure of entanglement in the ground state of a two-dimensional system with topological order. Specifically, we consider the toric code model \cite{Kitaev}, which can be solved exactly. The entanglement between a region $A$ and the rest of the system in the ground state of the toric code model has already been characterized previously using von Neumann \cite{Hamma} and Renyi \cite{Flammia} entropies \cite{Hamma,Flammia}. Here, we use the ability to compute the entanglement of mixed states to investigate its distribution in space, in the sense that we explicitly identify, within regions $A$ and $B$, the specific location of the entangled degrees of freedom.

We provide an analytical calculation of the negativity for a number of choices of regions $A$ and $B$, and use them to discriminate between two types of contributions to the entanglement of $\rho_{AB}$: \textit{boundary} entanglement and \textit{long-range} entanglement. Boundary entanglement entangles degrees of freedom that are close to the boundary between $A$ and $B$; it is proportional to the size of the boundary; and it includes the well-known topological term \cite{Hamma,TEE1,TEE2}. Long-range entanglement occurs only when the ground subspace is degenerate (e.g. on a torus), and its amount depends on the specific ground state of the system. While boundary entanglement survives the tracing out of bulk degrees of freedom inside regions $A$ and $B$, long-range entanglement is destroyed when a non-contractible regions in the interior of region $A$ or region $B$ is traced-out. The present decomposition of entanglement into different types of contributions is already implicitly present in previous papers, such as Ref. \cite{TEE1, Previous}. By studying the negativity, we can make this decomposition more concrete and explicitly identify the spatial origin of each contribution.

The rest of the paper is organized as follows. Section \ref{sect:background} reviews some background material. Sect. \ref{sect:contractible} analyzes the simple cases where region $A$ is contractible, in which case only boundary entanglement is present. Sect. \ref{sect:non-contractible} analyzes the case where both $A$ and $B$ are non-contractible regions of a torus, and shows the existence also of long-range entanglement.

\section{Background Material}
\label{sect:background}

In this section we briefly review the logarithmic negativity, the toric code model and its ground states, and describe the types of regions into which we will divide the lattice.

\subsection{Negativity}

The sum of negative eigenvalues of $\rho_{AB}^{T_A}$ can be seen to be equal to $(\tr |\rho_{AB}^{T_A}|-1)/2$, where $|O| \equiv \sqrt{O^{\dagger}O}$ \cite{VidalWerner}.

\begin{definition}
The \emph{logarithmic negativity} is
\begin{equation}\label{eq:negativity:definition}
E_\mathcal{N}^{A\mid B}\left(\rho_{AB}\right)
\equiv \log_2\tr\left|{\rho_{AB}}^{T_A}\right|.
\end{equation}
\end{definition}
If $E_\mathcal{N}^{A\mid B}\left(\rho_{AB}\right) > 0$, then $\rho_{AB}$ is entangled, as it follows from Peres criterion, Eq. \ref{eq:Peres}. It can be further seen \cite{VidalWerner} that the logarithmic negativity is additive, that is
\begin{equation}\label{eq:additivity}
E_{\mathcal{N}}(\varrho_{A_1B_1} \otimes \sigma_{A_2B_2}) = E_{\mathcal{N}}(\varrho_{A_1B_1}) + E_{\mathcal{N}}(\varrho_{A_2B_2}),
\end{equation}
and that for a pure state $\ket{\Psi}_{AB}$ it reduces to the Renyi entropy of index $q=1/2$,
\begin{equation}\label{eq:PureStates}
    E_\mathcal{N}(\proj{\Psi}) = S_{1/2}(\rho_{A}) = 2\log_2\left((\rho_A)^{1/2}\right).
\end{equation}
Also see appendix \ref{sec:appd:ext of log neg} for a general relation for Renyi entropies of any order.

\subsection{Toric code model}

To define the toric code model, we consider a generic graph (in practice a square lattice with various boundary conditions), and assign a two-level system or qubit to each edge.

\begin{figure}[t]
\includegraphics[width=8.0cm]{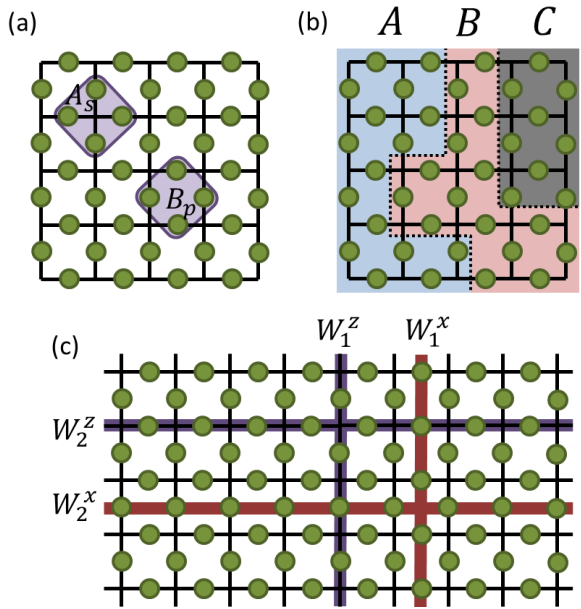}
\caption{
(a) Toric code on a lattice made of 40 sites, represented by circles. The support of a star term $A_s$ and of a plaquette term $B_p$, eq.\nobreakspace \textup {(\ref {eq:toric code:model:StarPlaquette})} is represented.
(b) Example of regions $A$, $B$, and $C$, such that no star or plaquette operator acts on more than two regions simultaneously.
(c) Square lattice with the topology of a torus, showing the support of loop operators $W_1^z$, $W_1^x$, $W_2^z$, and $W_2^x$.
}
\label{fig:ToricCode}
\end{figure}

\begin{definition}
For each vertex~$s$ and face~$p$ in the graph, we define the \emph{star}~$A_s$ and \emph{plaquette}~$B_p$ operators
\begin{equation}\label{eq:toric code:model:StarPlaquette}
A_s = \prod_{i\in\textnormal{legs}\left(s\right)} \sigma^x_i
\quad\quad\quad\quad\quad
B_p = \prod_{i\in\textnormal{boundary}\left(p\right)} \sigma^z_i
\end{equation}
where $\sigma_i$ denotes the Pauli operator acting on spin~$i$,  see Fig. \ref{fig:ToricCode}(a).
\end{definition}

The \emph{toric code} on a 2D graph with $n$ edges is then described by the Hilbert space~$\mathcal{H} = \left(\mathbb{C}^2\right)^{\otimes n}$ and by the local Hamiltonian
\begin{equation}
H=-U\sum_{s}A_s-J\sum_{p}B_p
\end{equation}
with $U>0$ and $J>0$, and where $A_s$~and~$B_p$ are operators acting on the edges ending at vertex $s$ and the edges surrounding face $p$, respectively.

\subsection{Ground states}

Given a complete set of $\kappa$ independent non-contractible loops on the graph [$\kappa = 2^\mathfrak{g}$ for a surface of genus $\mathfrak{g}$], we define the corresponding loop operators $\{W^z_i\}$, with $i=1, \cdots, \kappa$, where each $W_i^z$ is a loop of $\sigma^z$ operators. The ground state which is the $+1$ eigenstate of all non-contractible $\sigma^z$ loop operators~$W^z_i$ is then
\begin{equation}\label{eq:psi0}
\left|\psi_0\right\rangle
= N
\prod_{\substack{\textnormal{indep.}\\B_p}} \left(\frac{\identity+B_p}{2}\right)
\prod_{\substack{\textnormal{loop ops,}\\i=1}}^\kappa \left(\frac{\identity+W^z_i}{2}\right)
\left|+\cdots+\right\rangle,
\end{equation}
where $N$ is a normalization constant such that $\left\langle\psi_0\mid\psi_0\right\rangle = 1$.
This state is the equal weight superposition of all loop $\sigma^z$ operators on the lattice acting on the reference state~$\left|+\cdots+\right\rangle$.

Note that this ground state is identical to the more familiar form
\begin{equation}
\left|\psi_0\right\rangle
= N
\prod_{\substack{\textnormal{indep.}\\A_s}} \left(\frac{\identity+A_s}{2}\right)
\left|0\cdots0\right\rangle
.
\end{equation}

Other ground states can be obtained by acting on $\left|\psi_0\right\rangle$ with linear combinations of products of~$W^x_i$ operators. That is, a generic ground state can be written as
\begin{equation}\label{eq:generic}
\left|\psi\right\rangle
= \sum_{k_1,k_2,\cdots,k_\kappa=0}^1 c_{k_1 \cdots k_\kappa}
\left( W^x_1 \right)^{k_1} \left( W^x_2 \right)^{k_2} \cdots \left( W^x_\kappa \right)^{k_\kappa}
\left|\psi_0\right\rangle
\end{equation}
where $\sum_{\vec{k}} \left|c_{\vec{k}}\right|^2 = 1$ so that $\left\langle\psi\mid\psi\right\rangle = 1$.

\begin{example}
For a torus $\mathfrak{g}=1$, the ground state~$\left|\psi_0\right\rangle$ is
\begin{equation}
\left|\psi_0\right\rangle
= N
\prod_{\substack{\textnormal{indep.}\\B_p}} \left(\frac{\identity+B_p}{2}\right)
\left(\frac{\identity+W^z_1}{2}\right) \left(\frac{\identity+W^z_2}{2}\right)
\left|+\cdots+\right\rangle
\end{equation}
where the two loops $W^z_1$ and $W^z_2$ go around the two different radii of the torus, see Fig. \ref{fig:ToricCode}(b).

The general ground state can then be written as
\begin{equation} \label{eq:generic2}
\left|\psi\right\rangle
= \sum_{k_1,k_2=0}^1 c_{k_1 k_2}
\left( W^x_1 \right)^{k_1} \left( W^x_2 \right)^{k_2}
\left|\psi_0\right\rangle
\end{equation}
where $\left|c_{00}\right|^2 + \left|c_{01}\right|^2 + \left|c_{10}\right|^2 + \left|c_{11}\right|^2 = 1$.
\end{example}

In Sect. \ref{sect:non-contractible}, we will also use an alternative basis of ground states $\ket{\psi_i}$ on the torus, where $i=I,e,m,em$ is the anyonic flux threading the interior of the torus in the horizontal direction, see Eqs. \ref{eq:psiI}-\ref{eq:psiemem}.

\subsection{Regions}

In the next sections we will divide the lattice into two or more regions, see Fig. \ref{fig:ToricCode}(b). We will only consider divisions of the lattice into regions such that the star operators $A_s$ and plaquette operators $B_p$ only act non-trivially on at most two regions each. Although this restriction does not appear to be essential in order to perform an exact calculation of the negativity, it simplifies the derivation significantly.

\begin{figure}[t]
\includegraphics[width=8.6cm]{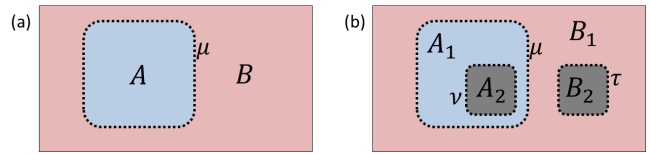}
\caption{
(a) Partition of a generic lattice (of arbitrary topology) into a contractible region $A$ and its complement $B$. Index $\mu$ labels loop configurations at the boundary between $A$ and $B$.
(b) Refined partition, in which $A$ further decomposes into $A_1$ and $A_2$, and $B$ into $B_1$ and $B_2$. Notice that $A_2$ shares no border with $B$, and $B_2$ shares no border with $A$. Indices $\mu$, $\nu$, and $\tau$ label loop configurations at the corresponding boundaries.
}
\label{fig:Regions1}
\end{figure}

\section{Region $A$ is contractible}
\label{sect:contractible}

In this section we analyze a setting where region $A$ is contractible, whereas region $B$ (the complement of $A$, such that $A \cup B$ is the whole lattice) is arbitrary, see Fig. \ref{fig:Regions1}(a). [Here we only discuss explicitly a contractible region $A$ that is simply connected, and later simply explain how the results generalize to an arbitrary number of simply connected, contractible regions]. First we compute the negativity for the pure state of $A\cup B$. Then, after considering a refined partition of the system into regions $A_1$, $A_2$, $B_1$, and $B_2$, where $A = A_1 \cup A_2$ and $B = B_1 \cup B_2$, see Fig. \ref{fig:Regions1}(b), we compute the negativity for the mixed state of $A_1\cup B_1$.

\subsection{Schmidt decomposition of a bipartition}

Consider the bipartition $A|B$ in Fig. \ref{fig:Regions1}(a), where region $A$ is contractible and $A\cup B$ is an arbitrary surface.

\begin{proposition}\label{prop:one region contr:schmidt decomposition:decomposition of GS1}
The ground state~$\left|\psi_0\right\rangle$ of Eq. \ref{eq:psi0} has the Schmidt decomposition
\begin{equation}\label{eq:one region contr:schmidt decomposition:decomposition of GS1}
\left|\psi_0\right\rangle
= {\mu_\textnormal{max}}^{-\frac{1}{2}} \sum_{\mu=1}^{\mu_\textnormal{max}}
\left|e_\mu\right\rangle_A \otimes \left|f_\mu\right\rangle_B
\end{equation}
where $\mu_\textnormal{max} = 2^{n_{AB}-1}$ such that $n_{AB}$ is the number of plaquettes spanning the boundary between~$A$ and~$B$.
Also, $\left\langle e_{\mu'} \mid e_\mu\right\rangle = \left\langle f_{\mu'} \mid f_\mu\right\rangle = \delta_{\mu\mu'}$.
\end{proposition}

\begin{example}
In the specific case of Fig. \ref{fig:Example}(a), on a square lattice, region $A$ is made of 10 sites, and its boundary is crossed by 8 plaquettes, so that $n_{AB} = 8$.
\end{example}

\begin{proof}
\begin{widetext}
The ground state~$\left|\psi_0\right\rangle$ in Eq. \ref{eq:psi0} can be written as
\begin{align}
\left|\psi_0\right\rangle
=&\; N
\left(\sum\textnormal{all loops of $\sigma^z$ operators in $A\cup B$}\right)
\left|+\cdots+\right\rangle	\nonumber \\
=&\; N
\left(\sum\textnormal{all loops of $\sigma^z$ operators crossing the boundary $A\mid B$ up to equivalance class}\right)	 \nonumber \\
&\cdot
\left(\sum\textnormal{all loops of $\sigma^z$ operators in $A$}\right)_A \left|+\cdots+\right\rangle_A
\otimes
\left(\sum\textnormal{all loops of $\sigma^z$ operators in $B$}\right)_B \left|+\cdots+\right\rangle_B
\end{align}
where the equivalence class is such that 2 loops are equivalent if one can be deformed into the other by loop operators in $A$~and/or~$B$ only.
\end{widetext}

\begin{figure}[t]
  \includegraphics[width=8.6cm]{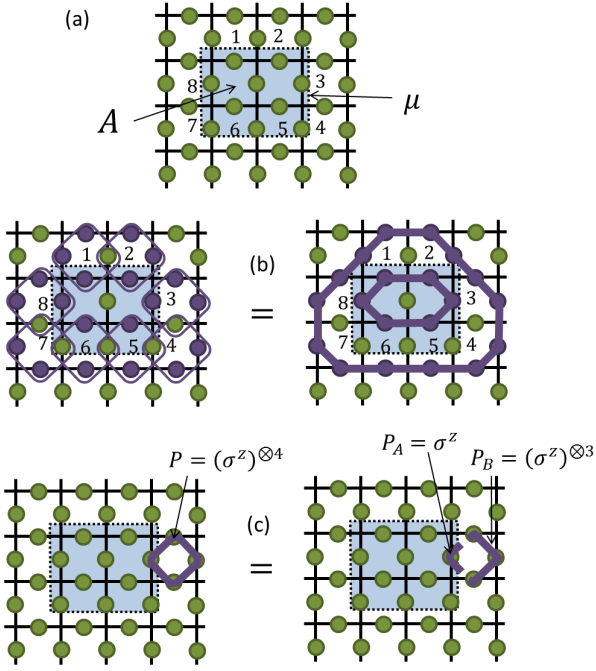}
\caption{
(a) Region $A$ made of 10 sites. There are 8 plaquette operators acting across its boundary.
(b) One of these plaquette operators can be obtained, plaquette operators contained inside of $A$ and by plaquette operators contained outside of $A$, as the product of the other 7 plaquette operators. Consequently, we say that there are $7$ independent plaquette operators across the boundary of region $A$.
(c) A plaquette operator $P$ acting across the boundary of $A$ can be broken into two operators $P_A$ and $P_B$, acting inside and outside of region $A$.
}
\label{fig:Example}
\end{figure}

Thus we can, without loss of generality, consider only loops made up of plaquette operators $P_{\mu}$ acting across the boundary.
\begin{align}
&\left|\psi_0\right\rangle	\nonumber \\
&=\; N
\left(\sum^{\mu_\textnormal{max}}_{\mu=1} \textnormal{prod of plaquette ops on boundary, }P_\mu\right)_{AB}	\nonumber \\
&\quad\quad \cdot
\left(\sum\textnormal{all loops in A}\right)_A \left|+\cdots+\right\rangle_A	\nonumber \\
&\quad\quad \otimes
\left(\sum\textnormal{all loops in B}\right)_B \left|+\cdots+\right\rangle_B
\end{align}

However, not all loops constructed from plaquettes on the boundary entangle regions~$A$ and~$B$.
In particular, the product of all plaquettes on the boundary~$A\mid B$ is a unitary operation acting locally in $A$ and~$B$ separately, see Fig. \ref{fig:Example}(b)
Therefore, the number of independent plaquettes on the boundary is $n_{AB}-1$ [This is valid for every boundary curve. In the present case there is only one such boundary curve]. Hence
\begin{equation}
\mu_\textnormal{max} = 2^{n_{AB}-1}
\end{equation}
where $n_{AB}$ is the total number of plaquettes on the boundary~$A\mid B$.

We can then `cut' the loops $P_\mu$ into 2 open strings in regions $A$~and~$B$ that meet on the boundary, see Fig. \ref{fig:Example}(c).
Formally, this is simply re-writing $P_\mu = {P_\mu}_A \otimes {P_\mu}_B$, so that
\begin{align}
\left|\psi_0\right\rangle
=&\; N
\sum^{\mu_\textnormal{max}}_{\mu=1}
\underbrace{\left[
	{P_\mu}_A \left(\sum\textnormal{all loops in A}\right)_A
	\left|+\cdots+\right\rangle_A
\right]}_{\propto\left|e_\mu\right\rangle_A}	\nonumber \\
&\otimes
\underbrace{\left[
	{P_\mu}_B \left(\sum\textnormal{all loops in B}\right)_B
	\left|+\cdots+\right\rangle_B
\right]}_{\propto\left|f_\mu\right\rangle_B}	\nonumber \\
=&\;
\sum^{\mu_\textnormal{max}}_{\mu=1} {\mu_\textnormal{max}}^{-\frac{1}{2}}
\left|e_\mu\right\rangle_A \otimes \left|f_\mu\right\rangle_B
\end{align}
where $\left|e_\mu\right\rangle$ and $\left|f_\mu\right\rangle$ are normalised so that $\left\langle e_{\mu'} \mid e_\mu\right\rangle = \left\langle f_{\mu'} \mid f_\mu\right\rangle = \delta_{\mu\mu'}$.
We thus have the desired result for the ground state~$\left|\psi_0\right\rangle$.
\end{proof}

\begin{proposition}\label{prop:one region contr:schmidt decomposition: decomposition of general GS}
A generic ground state~$\left|\psi\right\rangle$, Eq. \ref{eq:generic}, has the Schmidt decomposition
\begin{equation}
\left|\psi\right\rangle
= {\mu_\textnormal{max}}^{-\frac{1}{2}} \sum_{\mu=1}^{\mu_\textnormal{max}}
\left|e_\mu\right\rangle_A \otimes \left|f^{\left(\psi\right)}_\mu\right\rangle_B
\end{equation}
where $\mu_\textnormal{max} = 2^{n_{AB}-1}$ such that $n_{AB}$ is the number of plaquettes spanning the boundary~$A\mid B$.
Moreover, $\left\langle e_{\mu'} \mid e_\mu\right\rangle = \left\langle f^{\left(\psi\right)}_{\mu'} \mid f^{\left(\psi\right)}_\mu\right\rangle = \delta_{\mu\mu'}$.
\end{proposition}

\begin{proof}
Region~$A$ is contractible, that is, any non-contractible loop can be deformed locally so as to be entirely contained in region~$B$. Non-contractible loops of~$\sigma^x$ operators are precisely those required to obtain the general ground states from~$\left|\psi_0\right\rangle$, and hence the Schmidt coefficients of the general ground state are exactly the same as that of~$\left|\psi_0\right\rangle$.
We show this explicitly by combining Eqs. \ref{eq:generic} and \ref{eq:one region contr:schmidt decomposition:decomposition of GS1}, so that
\begin{align}
&\left|\psi\right\rangle	\nonumber \\
&= \sum_{k_1,k_2,\cdots,k_\kappa=0}^1 c_{k_1 k_2 \cdots k_\kappa}
\left( W^x_1 \right)^{k_1} \left( W^x_2 \right)^{k_2} \cdots \left( W^x_\kappa \right)^{k_\kappa}	\nonumber \\
&\quad\quad \cdot
\left( {\mu_\textnormal{max}}^{-1/2} \sum^{\mu_\textnormal{max}}_{\mu=1}
\left|e_\mu\right\rangle_A \otimes \left|f_\mu\right\rangle_B \right)	\nonumber \\
&=
{\mu_\textnormal{max}}^{-1/2} \sum^{\mu_\textnormal{max}}_{\mu=1}
\left|e_\mu\right\rangle_A	\nonumber \\
&\quad\quad \otimes
\underbrace{
	\left( \sum_{k_1,\cdots,k_\kappa=0}^1 c_{k_1 \cdots k_\kappa}
\left( W^x_1 \right)^{k_1} \cdots \left( W^x_\kappa \right)^{k_\kappa} \right)_B \left|f_\mu\right\rangle_B
}_{\left|f^{\left(\psi\right)}_\mu\right\rangle_B}
\end{align}
where
\begin{align}
&\left\langle f^{\left(\psi\right)}_{\mu'} \mid f^{\left(\psi\right)}_\mu \right\rangle_B	\nonumber \\
&= \left\langle f_{\mu'}\right|
\left( \sum_{k'} c_{k'} \prod_{i_{k'}} W^x_{i_{k'}} \right)^\dagger_B
\left( \sum_k c_k \prod_{i_k} W^x_{i_k} \right)_B
\left|f_\mu\right\rangle_B	\nonumber \\
&= \delta_{\mu,\mu'}
\end{align}
is orthonormal as desired.
\end{proof}

\subsection{Decomposition for a more refined partition}

Consider a more refined partition of spins (see \ref{fig:Regions1}, b),
where $A = A_1\cup A_2$ and~$B = B_1\cup B_2$ such that region $B_2$ is contractible.

\begin{proposition}\label{prop:one region contr:4-partite decomposition:decomposition of GS1}
The ground state $\left|\psi_0\right\rangle$ can be decomposed as
\begin{align}\label{eq:one region contr:4-partite decomposition:decomposition of GS1}
\left|\psi_0\right\rangle
=& {\mu_\textnormal{max}}^{-\frac{1}{2}} \sum_\mu^{\mu_\textnormal{max}}
\left( {\nu_\textnormal{max}}^{-\frac{1}{2}} \sum_\nu^{\nu_\textnormal{max}}
\left|e_{\mu\nu}\right\rangle_{A_1} \left|g_\nu\right\rangle_{A_2} \right)	
\nonumber \\
&\quad \cdot
\left( {\tau_\textnormal{max}}^{-\frac{1}{2}} \sum_\tau^{\tau_\textnormal{max}}
\left|f_{\mu\tau}\right\rangle_{B_1} \left|h_\tau\right\rangle_{B_2} \right)
\end{align}
where
$\left\langle e_{p'q'} \mid e_{pq}\right\rangle = \left\langle f_{p'q'} \mid f_{pq}\right\rangle = \delta_{pp'}\delta_{qq'}$
and
$\left\langle g_{p'} \mid g_{p}\right\rangle = \left\langle h_{p'} \mid h_{p}\right\rangle = \delta_{pp'}$.
Also,
\begin{align}
\mu_\textnormal{max} &= 2^{n_{A_1 B_1}-1}	\\
\nu_\textnormal{max} &= 2^{n_{A_1 A_2}-1}	\\
\tau_\textnormal{max} &= 2^{n_{B_1 B_2}-1}
.
\end{align}
\end{proposition}

\begin{remark}
Equation\nobreakspace \textup {(\ref {eq:one region contr:4-partite decomposition:decomposition of GS1})} is very much reminiscent of eq.\nobreakspace \textup {(\ref {eq:one region contr:schmidt decomposition:decomposition of GS1})}.
Indeed, we have that
\begin{equation}
\left|e_\mu\right\rangle_A
=
\left( {\nu_\textnormal{max}}^{-\frac{1}{2}} \sum_\nu^{\nu_\textnormal{max}}
\left|e_{\mu\nu}\right\rangle_{A_1} \left|g_\nu\right\rangle_{A_2} \right)	
\end{equation}
and
\begin{equation}
\left|f_\mu\right\rangle_B
=
\left( {\tau_\textnormal{max}}^{-\frac{1}{2}} \sum_\tau^{\tau_\textnormal{max}}
\left|f_{\mu\tau}\right\rangle_{B_1} \left|h_\tau\right\rangle_{B_2} \right)
\end{equation}
which we expect since eqs.\nobreakspace \textup {(\ref {eq:one region contr:schmidt decomposition:decomposition of GS1})} and\nobreakspace  \textup {(\ref {eq:one region contr:4-partite decomposition:decomposition of GS1})} differ only by a 'refining' of the partitions.
\end{remark}

\begin{proof}
See appendix.
\end{proof}

\begin{proposition}\label{prop:one region contr:4-partite decomposition:decomposition of general GS}
A generic ground state $\ket{\psi}$, Eq. \ref{eq:generic}, can be decomposed as
\begin{align}
\left|\psi\right\rangle
=& {\mu_\textnormal{max}}^{-\frac{1}{2}} \sum_\mu^{\mu_\textnormal{max}}
\left( {\nu_\textnormal{max}}^{-\frac{1}{2}} \sum_\nu^{\nu_\textnormal{max}}
\left|e_{\mu\nu}\right\rangle_{A_1} \left|g_\nu\right\rangle_{A_2} \right)	
\nonumber \\
&\quad \cdot
\left( {\tau_\textnormal{max}}^{-\frac{1}{2}} \sum_\tau^{\tau_\textnormal{max}}
\left|f^{\left(\psi\right)}_{\mu\tau}\right\rangle_{B_1} \left|h_\tau\right\rangle_{B_2} \right)
\end{align}
where
$\left\langle e_{p'q'} \mid e_{pq}\right\rangle = \left\langle f^{\left(\psi\right)}_{p'q'} \mid f^{\left(\psi\right)}_{pq}\right\rangle = \delta_{pp'}\delta_{qq'}$
and
$\left\langle g_{p'} \mid g_{p}\right\rangle = \left\langle h_{p'} \mid h_{p}\right\rangle = \delta_{pp'}$, and where
$\mu_\textnormal{max}$, $\nu_\textnormal{max}$, $\tau_\textnormal{max}$ are defined as in proposition\nobreakspace \ref {prop:one region contr:4-partite decomposition:decomposition of GS1}.
\end{proposition}

\begin{proof}
As with proposition\nobreakspace \ref {prop:one region contr:schmidt decomposition: decomposition of general GS}, to get the result for any other ground state, we note that any non-contractible loop can be entirely contained in region~$B_1$.
Hence we get the desired result by combining the previous equation with Eq. \ref{eq:generic}.
\end{proof}

\subsection{Entanglement Negativity}

\subsubsection{Entanglement between $A$ and $B$}

\begin{proposition}\label{prop:one region contr:negativity:negativity of bipartition}
The logarithmic negativity~$E_\mathcal{N}^{A\mid B}$ between regions $A$ and~$B$ for any ground state is given by
\begin{equation}
E_\mathcal{N}^{A\mid B}
\left( \left|\psi\right\rangle\left\langle\psi\right| \right)
= n_{AB} - 1
.
\end{equation}
\end{proposition}

\begin{remark}
Since in this case the state of $A\cup B$ is pure, this expression for the logarithmic negativity is equivalent to the Renyi entropy $S_{1/2}$, Eq. \ref{eq:PureStates}. Thus, the above proposition recovers the result of previous calculations of ground-state entanglement in the toric code, see e.g. Refs. \cite{Hamma,Flammia}.
\end{remark}

\begin{proof}
From proposition\nobreakspace \ref {prop:one region contr:schmidt decomposition: decomposition of general GS},
\begin{equation}
\left|\psi\right\rangle
= 2^{-\left(n_{AB}-1\right)/2} \sum_{\mu=1}^{\mu_\textnormal{max}}
\left|e_\mu\right\rangle_A \otimes \left|f^{\left(\psi\right)}_\mu\right\rangle_B
.
\end{equation}
Let $\rho_{AB} = \left|\psi\right\rangle\left\langle\psi\right|$. Then,
\begin{equation}
{\rho_{AB}}^{T_A}
= 2^{-\left(n_{AB}-1\right)}
\sum_{\mu,\mu'=1}^{\mu_\textnormal{max}}
\left|e_\mu\right\rangle\left\langle e_{\mu'}\right|_A
\otimes
\left|f^{\left(\psi\right)}_{\mu'}\right\rangle\left\langle f^{\left(\psi\right)}_\mu\right|_B,
\end{equation}
where we choose to take the partial transposition in the $\ket{e_{\mu}}_A$ basis. Taking the square, we have that
\begin{equation}
\left(\rho_{AB}^{T_A} \right)^2
= 2^{-2\left(n_{AB}-1\right)}
\sum_{\mu,\mu'=1}^{\mu_\textnormal{max}}
\left|e_\mu\right\rangle\left\langle e_\mu\right|_A
\otimes
\left|f^{\left(\psi\right)}_{\mu'}\right\rangle\left\langle f^{\left(\psi\right)}_{\mu'}\right|_B.
\end{equation}
Thus we can compute
\begin{align}
&\left| {\rho_{AB}}^{T_A} \right|
= \sqrt{ {{\rho_{AB}}^{T_A}}^\dagger {\rho_{AB}}^{T_A} }
= \sqrt{ \left( {\rho_{AB}}^{T_A} \right)^2 }	\nonumber \\
&\quad\quad
= \sqrt{ 2^{-2\left(n_{AB}-1\right)} }
\sum_{\mu,\mu'=1}^{\mu_\textnormal{max}}
\left|e_\mu\right\rangle\left\langle e_\mu\right|_A
\otimes
\left|f^{\left(\psi\right)}_{\mu'}\right\rangle\left\langle f^{\left(\psi\right)}_{\mu'}\right|_B
\end{align}
where the last equality follows since the state is already in its eigenvalue decomposition.
Therefore we have the negativity as
\begin{align}
&E_\mathcal{N}^{A\mid B}
\left( \left|\psi\right\rangle\left\langle\psi\right| \right)	\nonumber \\
&= \log_2\tr \left| {\rho_{AB}}^{T_A} \right|	\nonumber \\
&= \log_2\tr \left(
2^{-\left(n_{AB}-1\right)}
\sum_{\mu,\mu'=1}^{\mu_\textnormal{max}}
\left|e_\mu\right\rangle\left\langle e_\mu\right|_A
\otimes
\left|f^{\left(\psi\right)}_{\mu'}\right\rangle\left\langle f^{\left(\psi\right)}_{\mu'}\right|_B
\right)	\nonumber \\
&= \log_2 2^{\left(n_{AB}-1\right)}
= n_{AB}-1
\end{align}
as required.
\end{proof}

\begin{remark}
Since the state $\rho_{AB}$ is a pure state, the negativity for the bipartition is exactly the same as the R\'{e}nyi entropy of $\rho_A = tr_B\left(\rho_{AB}\right)$, as expected.
\end{remark}

\subsubsection{Entanglement between $A_1$ and $B_1$}

\begin{proposition}
The logarithmic negativity~$E_\mathcal{N}^{A_1\mid B_1}$ between regions $A_1$ and~$B_1$ for a generic ground state $\ket{\psi}$ is given by
\begin{equation}
E_\mathcal{N}^{A_1\mid B_1} \left( \rho_{A_1 B_1} \right)
= n_{AB} - 1
\end{equation}
where $\rho_{A_1 B_1} = \tr_{A_2 B_2} \left|\psi\right\rangle\left\langle\psi\right|$.
\end{proposition}

\begin{proof}
From proposition\nobreakspace \ref {prop:one region contr:4-partite decomposition:decomposition of general GS},
\begin{align}
\left|\psi\right\rangle
=& {\mu_\textnormal{max}}^{-\frac{1}{2}} \sum_\mu^{\mu_\textnormal{max}}
\left( {\nu_\textnormal{max}}^{-\frac{1}{2}} \sum_\nu^{\nu_\textnormal{max}}
\left|e_{\mu\nu}\right\rangle_{A_1} \left|g_\nu\right\rangle_{A_2} \right)	
\nonumber \\
&\quad \cdot
\left( {\tau_\textnormal{max}}^{-\frac{1}{2}} \sum_\tau^{\tau_\textnormal{max}}
\left|f^{\left(\psi\right)}_{\mu\tau}\right\rangle_{B_1} \left|h_\tau\right\rangle_{B_2} \right)
.
\end{align}
So, the reduced density matrix over~$A_1$ and~$B_1$ is
\begin{align}
\rho_{A_1 B_1}
&= \tr_{A_2 B_2} \left|\psi\right\rangle\left\langle\psi\right|	\nonumber \\
&=
{\mu_\textnormal{max}}^{-1} \sum_{\mu,\mu'=1}^{\mu_\textnormal{max}}
\cdot
{\nu_\textnormal{max}}^{-1} \sum_{\nu=1}^{\nu_\textnormal{max}}
\left|e_{\mu'\nu}\right\rangle\left\langle e_{\mu\nu}\right|_A	\nonumber \\
&\quad\quad \otimes
{\tau_\textnormal{max}}^{-1} \sum_{\tau=1}^{\tau_\textnormal{max}}
\left|f^{\left(\psi\right)}_{\mu'\tau}\right\rangle\left\langle f^{\left(\psi\right)}_{\mu\tau}\right|_B
.
\end{align}
Taking the partial transpose and squaring, we have that
\begin{align}
\left( \rho_{A_1 B_1}^{T_A} \right)^2
&=
{\mu_\textnormal{max}}^{-2}
\sum_{\mu=1}^{\mu_\textnormal{max}}
\sum_{\nu=1}^{\nu_\textnormal{max}} {\nu_\textnormal{max}}^{-2}
\left|e_{\mu\nu}\right\rangle\left\langle e_{\mu\nu}\right|_A	\nonumber \\
&\quad\quad \otimes
\sum_{\mu'=1}^{\mu_\textnormal{max}}
\sum_{\tau=1}^{\tau_\textnormal{max}} {\tau_\textnormal{max}}^{-2}
\left|f^{\left(\psi\right)}_{\mu'\tau}\right\rangle\left\langle f^{\left(\psi\right)}_{\mu'\tau}\right|_B
\end{align}
and therefore we have that
\begin{align}
\left| {\rho_{A_1 B_1}}^{T_A} \right|
&= \sqrt{ {{\rho_{A_1 B_1}}^{T_A}}^\dagger {\rho_{A_1 B_1}}^{T_A} }
= \sqrt{ \left( {\rho_{A_1 B_1}}^{T_A} \right)^2 }	\nonumber \\
&=
{\mu_\textnormal{max}}^{-1}
\sum_{\mu=1}^{\mu_\textnormal{max}}
\sum_{\nu=1}^{\nu_\textnormal{max}} {\nu_\textnormal{max}}^{-1}
\left|e_{\mu\nu}\right\rangle\left\langle e_{\mu\nu}\right|_A	\nonumber \\
&\quad\quad \otimes
\sum_{\mu'=1}^{\mu_\textnormal{max}}
\sum_{\tau=1}^{\tau_\textnormal{max}} {\tau_\textnormal{max}}^{-1}
\left|f^{\left(\psi\right)}_{\mu'\tau}\right\rangle\left\langle f^{\left(\psi\right)}_{\mu'\tau}\right|_B
\end{align}
where the last equality follows because the state is, as with the previous case, already in its eigenvalue decomposition.
Therefore the negativity reads
\begin{align}
&E_\mathcal{N}^{A_1\mid B_1} \left( \rho_{A_1 B_1} \right)	\nonumber \\
&= \log_2\tr \left| {\rho_{A_1 B_1}}^{T_A} \right|	\nonumber \\
&= \log_2\tr \left(
{\mu_\textnormal{max}}^{-1}
\sum_{\mu=1}^{\mu_\textnormal{max}}
\sum_{\nu=1}^{\nu_\textnormal{max}} {\nu_\textnormal{max}}^{-1}
\left|e_{\mu\nu}\right\rangle\left\langle e_{\mu\nu}\right|_A	\right.\nonumber \\
&\left.\quad\quad \otimes
\sum_{\mu'=1}^{\mu_\textnormal{max}}
\sum_{\tau=1}^{\tau_\textnormal{max}} {\tau_\textnormal{max}}^{-1}
\left|f^{\left(\psi\right)}_{\mu'\tau}\right\rangle\left\langle f^{\left(\psi\right)}_{\mu'\tau}\right|_B
\right)
	\nonumber \\
&= \log_2 {\mu_\textnormal{max}}
= n_{AB}-1
\end{align}
as required.
\end{proof}

\subsection{Interpretation}

We have just seen that the negativity of $\rho_{AB}$ (pure state) and of $\rho_{A_1B_1}$ (mixed state) are the same,
\begin{equation}\label{eq:interpretation}
    E_{\mathcal{N}}^{A|B}(\rho_{AB}) = E_{\mathcal{N}}^{A_1|B_1} (\rho_{A_1B_1}) = n_{AB}-1.
\end{equation}
This result indicates that the entanglement between parts $A$ and $B$, as measured by the negativity, survives the operation of tracing out the bulk of $A$ and $B$ (that is, tracing out regions $A_2$ and $B_2$). We conclude that this entanglement must be entangling the degrees of freedom of $A$ and $B$ that are very close to the boundary between these two regions. We therefore refer to this entanglement as \textit{boundary} entanglement. Notice also that boundary entanglement is proportional to the size of the boundary (in this case, as measured by the number of plaquette terms across the boundary), and contains a universal correction $-1$, the topological entanglement entropy \cite{Hamma,TEE1,TEE2}. [We can relate the $-1$ in Eq. \ref{eq:interpretation} to the topological entanglement entropy thanks to the fact that for pure states the negativity is the Renyi entropy $S_{1/2}$, Eq. \ref{eq:PureStates}].

For a contractible region $A$ made of $p$ simply connected subregions $\alpha=1, \cdots, p$, a generalization of the above calculations shows that the pure-state negativity is made of $p$ contributions,
\begin{equation}\label{eq:transpose2}
    E_{\mathcal{N}}^{A|B} = \sum_{\alpha} (n_{AB}^{(\alpha)} - 1),
\end{equation}
where $n_{AB}^{(\alpha)}$ is the number of plaquettes across the boundary of subregion $\alpha$. All of these contributions are robust against the tracing out of bulk degrees of freedom. A given contribution will disappear upon tracing out the corresponding subregion.

\section{Regions $A$ and $B$ are not contractible}

\label{sect:non-contractible}

In this section we specialize, for the sake of concreteness, to a lattice $A \cup B$ with the topology of a torus, and consider non-contractible regions $A$ and $B$ connected by two boundaries, see Fig. \ref{fig:Regions2}(a). First we compute the negativity for the pure state of $A\cup B$. Then, once again, we consider a refined partition into regions $A_1$, $A_2$, $B_1$, and $B_2$, where $A = A_1 \cup A_2$ and $B = B_1 \cup B_2$, see Fig. \ref{fig:Regions2}(b), and compute the negativity for the mixed states of $A\cup B_1$ and $A_1\cup B_1$.

\begin{figure}[t]
  \includegraphics[width=8.6cm]{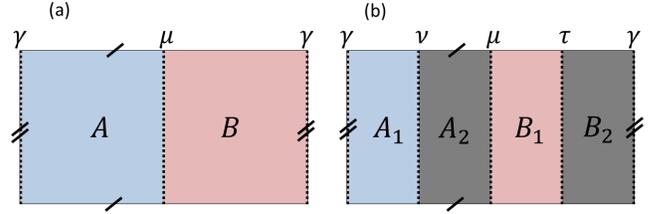}
\caption{
(a) Partition of a torus (lattice with periodic boundary conditions in the horizontal and vertical directions) into regions $A$ and $B$. Notice that regions $A$ and $B$ and not contractible share two boundaries. Indices $\mu$ and $\nu$ are used to label loop configurations at these boundaries.
(b) Refined partition of the torus. Each subregion $A_1$, $A_2$, $B_1$, and $B_2$ is non-contractible.
}
\label{fig:Regions2}
\end{figure}

\subsection{Schmidt decomposition of a bipartition}

Consider the bipartition $A|B$ of a torus in Fig. \ref{fig:Regions2}(a).
Let $\left|\psi_{\identity}\right\rangle$ be the ground state that is the $+1$~eigenstate of~$W^z_1$ and~$W^x_1$, where the `$1$'-direction is vertical in fig.\nobreakspace \ref {fig:Regions2}(c).
\begin{equation} \label{eq:psiI}
\left|\psi_{\identity}\right\rangle
= N
\prod_{\substack{\textnormal{indep.}\\B_p}} \left(\frac{\identity+B_p}{2}\right)
\left(\frac{\identity+W^z_1}{2}\right)
\left|+\cdots+\right\rangle
\end{equation}
where $N$ is the normalization such that $\left\langle\psi_{\identity}\mid\psi_{\identity}\right\rangle = 1$.

We also label the other  eigenstates of~$W^z_1$ and~$W^x_1$ as
\begin{eqnarray}
\left|\psi_e\right\rangle &=& W^x_2 \left|\psi_{\identity}\right\rangle \label{eq:psie}\\
\left|\psi_m\right\rangle &=& W^z_2 \left|\psi_{\identity}\right\rangle	\\
\left|\psi_{em}\right\rangle &=& W^x_2 W^z_2 \left|\psi_{\identity}\right\rangle \label{eq:psiemem}
\end{eqnarray}
The states $\left\lbrace \left|\psi_{\identity}\right\rangle, \left|\psi_e\right\rangle, \left|\psi_m\right\rangle, \left|\psi_{em}\right\rangle \right\rbrace$ have definite anyon flux (of type $i = I,e,m,em$) threading through the interior of the torus in the horizontal direction. Since this set forms a complete basis of the ground state space, a generic ground state~$\left|\psi\right\rangle$ can be written as
\begin{equation}\label{eq:more than 1 region non-contr:general GS}
\left|\psi\right\rangle
= \sum_{i=\identity,e,m,em} c_i \left|\psi_i\right\rangle
\end{equation}

\begin{proposition}\label{prop:more than 1 region non-contr:bipartition:decomposition of GSident}
The ground state~$\left|\psi_{\identity}\right\rangle$ has the Schmidt decomposition
\begin{equation}
\left|\psi_{\identity}\right\rangle
= {\mu_\textnormal{max}}^{-\frac{1}{2}} {\gamma_\textnormal{max}}^{-\frac{1}{2}}
\sum_{\mu=1}^{\mu_\textnormal{max}} \sum_{\gamma=1}^{\gamma_\textnormal{max}}
\left|e_{\mu\gamma}\right\rangle_A \otimes \left|f_{\mu\gamma}\right\rangle_B
\end{equation}
where $\mu_\textnormal{max} = 2^{n_{AB}^{\left(1\right)}-1}$ such that $n_{AB}^{\left(1\right)}$ is the number of plaquettes spanning one of the common boundaries of $A$ and $B$, and
where $\gamma_\textnormal{max} = 2^{n_{AB}^{\left(2\right)}-1}$ such that $n_{AB}^{\left(2\right)}$ is the number of plaquettes spanning the other boundary of $A$ and $B$.
Also, $\left\langle e_{\mu'\gamma'} \mid e_{\mu\gamma}\right\rangle
= \left\langle f_{\mu'\gamma'} \mid f_{\mu\gamma}\right\rangle
= \delta_{\mu\mu'}\delta_{\gamma\gamma'}$.
\end{proposition}

\begin{proof}
See appendix.
\end{proof}

\begin{proposition}\label{prop:more than 1 region non-contr:bipartition:decomposition of GS(i)}
Let $\left|\psi_i\right\rangle$ with $i = \identity,e,m,em$ denote a ground state which is an eigenstate of the operators~$W^z_1$ and~$W^x_1$ (i.e. has definite flux).
The ground state~$\left|\psi_i\right\rangle$ has the Schmidt decomposition
\begin{equation}
\left|\psi_i\right\rangle
= {\mu_\textnormal{max}}^{-\frac{1}{2}} {\gamma_\textnormal{max}}^{-\frac{1}{2}}
\sum_{\mu=1}^{\mu_\textnormal{max}} \sum_{\gamma=1}^{\gamma_\textnormal{max}}
\left|e^{\left(\psi_i\right)}_{\mu\gamma}\right\rangle_A
\otimes \left|f^{\left(\psi_i\right)}_{\mu\gamma}\right\rangle_B
\end{equation}
where $\mu_\textnormal{max}$ and $\gamma_\textnormal{max}$ are defined in proposition\nobreakspace \ref {prop:more than 1 region non-contr:bipartition:decomposition of GSident}.
Also, $\left\langle e^{\left(\psi_{i'}\right)}_{\mu'\gamma'} \mid e^{\left(\psi_i\right)}_{\mu\gamma}\right\rangle
= \left\langle f^{\left(\psi_{i'}\right)}_{\mu'\gamma'} \mid f^{\left(\psi_i\right)}_{\mu\gamma}\right\rangle
= \delta_{\mu\mu'}\delta_{\gamma\gamma'}\delta_{ii'}$.
\end{proposition}

\begin{proof}
The case for $i=\identity$ is covered in proposition\nobreakspace \ref {prop:more than 1 region non-contr:bipartition:decomposition of GSident}.
For $i=e$, from Eq. \ref{eq:psie} and the $i=\identity$ case we have
\begin{align}
&\left|\psi_e\right\rangle	\nonumber \\
&=
{\mu_\textnormal{max}}^{-\frac{1}{2}} {\gamma_\textnormal{max}}^{-\frac{1}{2}}
W^x_2
\sum_{\mu=1}^{\mu_\textnormal{max}} \sum_{\gamma=1}^{\gamma_\textnormal{max}}
\left|e_{\mu\gamma}\right\rangle_A \otimes \left|f_{\mu\gamma}\right\rangle_B	\nonumber \\
&=
{\mu_\textnormal{max}}^{-\frac{1}{2}} {\gamma_\textnormal{max}}^{-\frac{1}{2}}
\sum_{\mu=1}^{\mu_\textnormal{max}} \sum_{\gamma=1}^{\gamma_\textnormal{max}}
\underbrace{
	{W^x_2}_A \left|e_{\mu\gamma}\right\rangle_A
}_{\left|e^{\left(\psi_e\right)}_{\mu\gamma}\right\rangle_A}
\otimes
\underbrace{
	{W^x_2}_B \left|f_{\mu\gamma}\right\rangle_B
}_{\left|f^{\left(\psi_e\right)}_{\mu\gamma}\right\rangle_B}
\end{align}
where we have rewritten $W^x_2 = {W^x_2}_A \otimes {W^x_2}_B$.
Here, $\left\langle e^{\left(\psi_{\identity}\right)}_{\mu'\gamma'} \mid e^{\left(\psi_e\right)}_{\mu\gamma}\right\rangle
= \left\langle f^{\left(\psi_{\identity}\right)}_{\mu'\gamma'} \mid f^{\left(\psi_e\right)}_{\mu\gamma}\right\rangle
= 0$
and
$\left\langle e^{\left(\psi_e\right)}_{\mu'\gamma'} \mid e^{\left(\psi_e\right)}_{\mu\gamma}\right\rangle
= \left\langle f^{\left(\psi_e\right)}_{\mu'\gamma'} \mid f^{\left(\psi_e\right)}_{\mu\gamma}\right\rangle
= \delta_{\mu\mu'}\delta_{\gamma\gamma'}$.
The cases $i=m,em$ follow in an analogous manner.
\end{proof}

\begin{proposition}\label{prop:more than 1 region non-contr:bipartition:decomposition of general GS}
The general ground state~$\left|\psi\right\rangle$ has the Schmidt decomposition
\begin{align}
\left|\psi\right\rangle
= \sum_{i=\identity,e,m,em} c_i
\sum_{\mu=1}^{\mu_\textnormal{max}} \sum_{\gamma=1}^{\gamma_\textnormal{max}}
&{\mu_\textnormal{max}}^{-\frac{1}{2}} {\gamma_\textnormal{max}}^{-\frac{1}{2}}	\nonumber \\
&\quad\quad\quad \cdot
\left|e^{\left(\psi_i\right)}_{\mu\gamma}\right\rangle_A
\otimes \left|f^{\left(\psi_i\right)}_{\mu\gamma}\right\rangle_B
\end{align}
where $\mu_\textnormal{max}$ and $\gamma_\textnormal{max}$ are defined in proposition\nobreakspace \ref {prop:more than 1 region non-contr:bipartition:decomposition of GSident}, and
$c_i$ is defined in eq.\nobreakspace \textup {(\ref {eq:more than 1 region non-contr:general GS})}.
\end{proposition}

\begin{proof}
It follows directly from proposition\nobreakspace \ref {prop:more than 1 region non-contr:bipartition:decomposition of GS(i)} and the fact that
$\left\langle e^{\left(\psi_{i'}\right)}_{\mu'\gamma'} \mid e^{\left(\psi_i\right)}_{\mu\gamma}\right\rangle
= \left\langle f^{\left(\psi_{i'}\right)}_{\mu'\gamma'} \mid f^{\left(\psi_i\right)}_{\mu\gamma}\right\rangle
= \delta_{\mu\mu'}\delta_{\gamma\gamma'}\delta_{ii'}$.
\end{proof}

\subsection{Decomposition for a more refined partition}

Consider the more refined partition of Fig. \ref{fig:Regions2}(b), where $A = A_1\cup A_2$ and~$B = B_1\cup B_2$.

\begin{proposition}\label{prop:more than 1 region non-contr:refined partition:decomposition of GSident}
The ground state~$\left|\psi_{\identity}\right\rangle$ has the decomposition
\begin{align}
\left|\psi_{\identity}\right\rangle
=& {\mu_\textnormal{max}}^{-\frac{1}{2}} {\gamma_\textnormal{max}}^{-\frac{1}{2}}
{\nu_\textnormal{max}}^{-\frac{1}{2}} {\tau_\textnormal{max}}^{-\frac{1}{2}}	\nonumber \\
&
\sum_{\mu=1}^{\mu_\textnormal{max}} \sum_{\gamma=1}^{\gamma_\textnormal{max}}
\sum_{\nu=1}^{\nu_\textnormal{max}} \sum_{\tau=1}^{\tau_\textnormal{max}}
\left|e_{\nu\gamma}\right\rangle_{A_1} \left|f_{\mu\tau}\right\rangle_{B_1}
\left|g_{\mu\nu}\right\rangle_{A_2} \left|h_{\tau\gamma}\right\rangle_{B_2}
\end{align}
where the indices $\mu$, $\gamma$, $\nu$ and $\tau$ range from~$1$ to
\begin{subequations}
\begin{align}
\mu_\textnormal{max} &= 2^{n_{A_2 B_1}-1}	\\
\gamma_\textnormal{max} &= 2^{n_{A_1 B_2}-1}	\\
\nu_\textnormal{max} &= 2^{n_{A_1 A_2}-1}	\\
\tau_\textnormal{max} &= 2^{n_{B_1 B_2}-1}
.
\end{align}
\end{subequations}
As usual, $n$ denotes the number of plaquettes spanning one of the relevant boundary, and the states satisfy
$\left\langle e_{p'q'} \mid e_{pq}\right\rangle
= \left\langle f_{p'q'} \mid f_{pq}\right\rangle
= \left\langle g_{p'q'} \mid g_{pq}\right\rangle
= \left\langle h_{p'q'} \mid h_{pq}\right\rangle
= \delta_{pp'}\delta_{qq'}$.
\end{proposition}

\begin{proof}
Directly analogous to the proof for proposition\nobreakspace \ref {prop:more than 1 region non-contr:bipartition:decomposition of GSident}.
\end{proof}

\begin{widetext}
\begin{proposition}\label{prop:more than 1 region non-contr:refined partition:decomposition of general GS}
A generic ground state~$\left|\psi\right\rangle$ has the decomposition
\begin{equation}
\left|\psi\right\rangle
= \sum_{i=\identity,e,m,em} c_i
{\mu_\textnormal{max}}^{-\frac{1}{2}} {\gamma_\textnormal{max}}^{-\frac{1}{2}}
{\nu_\textnormal{max}}^{-\frac{1}{2}} {\tau_\textnormal{max}}^{-\frac{1}{2}}
\sum_{\mu=1}^{\mu_\textnormal{max}} \sum_{\gamma=1}^{\gamma_\textnormal{max}}
\sum_{\nu=1}^{\nu_\textnormal{max}} \sum_{\tau=1}^{\tau_\textnormal{max}}
\left|e^{\left(\psi_i\right)}_{\nu\gamma}\right\rangle_{A_1}
\left|f^{\left(\psi_i\right)}_{\mu\tau}\right\rangle_{B_1}
\left|g^{\left(\psi_i\right)}_{\mu\nu}\right\rangle_{A_2}
\left|h^{\left(\psi_i\right)}_{\tau\gamma}\right\rangle_{B_2}
\end{equation}
where $\mu_\textnormal{max}$, $\gamma_\textnormal{max}$, $\nu_\textnormal{max}$ and $\tau_\textnormal{max}$ are defined in proposition\nobreakspace \ref {prop:more than 1 region non-contr:refined partition:decomposition of GSident}, and
$c_i$ is defined in eq.\nobreakspace \textup {(\ref {eq:more than 1 region non-contr:general GS})}.
Moreover, the states satisfy
$\left\langle e^{\left(\psi_{i'}\right)}_{p'q'} \mid e^{\left(\psi_i\right)}_{pq} \right\rangle
= \left\langle f^{\left(\psi_{i'}\right)}_{p'q'} \mid f^{\left(\psi_i\right)}_{pq} \right\rangle
= \left\langle g^{\left(\psi_{i'}\right)}_{p'q'} \mid g^{\left(\psi_i\right)}_{pq} \right\rangle
= \left\langle h^{\left(\psi_{i'}\right)}_{p'q'} \mid h^{\left(\psi_i\right)}_{pq} \right\rangle
= \delta_{pp'}\delta_{qq'}\delta_{ii'}$.
\end{proposition}
\end{widetext}

\begin{proof}
Directly analogous to the proof for propositions\nobreakspace \ref {prop:more than 1 region non-contr:bipartition:decomposition of GS(i)} and\nobreakspace  \ref {prop:more than 1 region non-contr:bipartition:decomposition of general GS}.
\end{proof}

\subsection{Entanglement negativity}

\subsubsection{Entanglement between $A$ and $B$}

\begin{proposition}
The logarithmic negativity~$E_\mathcal{N}^{A\mid B}$ between regions $A$ and~$B$ for any ground state with definite flux~$\left|\psi_i\right\rangle$ is
\begin{equation}\label{eq:more than 1 region non-contr:bipartition:negativity for GS(i)}
E_\mathcal{N}^{A\mid B}
\left( \left|\psi_i\right\rangle\left\langle\psi_i\right| \right)
= \left( n_{AB}^{\left(1\right)} - 1 \right)
+ \left( n_{AB}^{\left(2\right)} - 1 \right)
\end{equation}
for any $i = \identity,e,m,em$.
\end{proposition}

\begin{proof}
From the Schmidt decomposition of $\left|\psi_i\right\rangle$ in proposition\nobreakspace \ref {prop:more than 1 region non-contr:bipartition:decomposition of GS(i)}, and
following the same procedure as outlined in the proof of proposition\nobreakspace \ref {prop:one region contr:negativity:negativity of bipartition},
we have the negativity as
\begin{equation}
E_\mathcal{N}^{A\mid B}
= \log_2\left( \mu_\textnormal{max} \gamma_\textnormal{max} \right)
.
\end{equation}
The desired result follows directly from the definition of~$\mu_\textnormal{max}$ and~$\gamma_\textnormal{max}$ (proposition\nobreakspace \ref {prop:more than 1 region non-contr:bipartition:decomposition of GSident}).
\end{proof}

\begin{proposition}
The logarithmic negativity~$E_\mathcal{N}^{A\mid B}$ between regions $A$ and~$B$ for any ground state~$\left|\psi\right\rangle$ is
\begin{equation}
E_\mathcal{N}^{A\mid B}
\left( \left|\psi\right\rangle\left\langle\psi\right| \right)
= 2 \log_2\sum_i \left|c_i\right|
+ \left( n_{AB}^{\left(1\right)} - 1 \right)
+ \left( n_{AB}^{\left(2\right)} - 1 \right)
\end{equation}
for any $i = \identity,e,m,em$.
\end{proposition}

\begin{proof}
From the decomposition of $\left|\psi\right\rangle$ in proposition\nobreakspace \ref {prop:more than 1 region non-contr:bipartition:decomposition of general GS}, and
following the previous proof,
we have the negativity as
\begin{equation}
E_\mathcal{N}^{A\mid B}
= \log_2\left[ \left(\sum_{i,j}\left|c_i\right|\left|c_j\right|\right)
\mu_\textnormal{max} \gamma_\textnormal{max} \right]
.
\end{equation}
The desired result follows directly from the observation that
$\sum_{i,j}\left|c_i\right|\left|c_j\right| = \left(\sum_i\left|c_i\right|\right)^2$.
\end{proof}

\begin{remark}
The logarithmic negativity~$E_\mathcal{N}^{A\mid B}$ for any ground state~$\left|\psi\right\rangle$ is a combination of two entropies:
\begin{align}
&E_\mathcal{N}^{A\mid B}
\left( \left|\psi\right\rangle\left\langle\psi\right| \right)	\nonumber \\
&=
\underbrace{
	2 \log_2\sum_i \left|c_i\right|
}_{\textnormal{new contribution}}
+ \underbrace{
	\left( n_{AB}^{\left(1\right)} - 1 \right)
	+ \left( n_{AB}^{\left(2\right)} - 1 \right)
}_{\textnormal{same as eq.\nobreakspace \textup {(\ref {eq:more than 1 region non-contr:bipartition:negativity for GS(i)})}}}
.
\end{align}
\end{remark}

\subsubsection{Entanglement between $A$ and $B_1$}

\begin{proposition}\label{prop:more than 1 region non-contr:refined partition:negativity:ent between A and B1}
The logarithmic negativity~$E_\mathcal{N}^{A\mid B_1}$ between regions $A$ and~$B_1$ for any ground state is given by
\begin{equation}
E_\mathcal{N}^{A\mid B_1} \left( \rho_{A B_1} \right)
= n_{A_2 B_1} - 1
\end{equation}
where $\rho_{A B_1} = \tr_{B_2} \left|\psi\right\rangle\left\langle\psi\right|$.
\end{proposition}

\begin{proof}
See appendix.
\end{proof}

\subsubsection{Entanglement between $A_1$ and $B_1$}

\begin{proposition}
The logarithmic negativity~$E_\mathcal{N}^{A_1\mid B_1}$ between regions $A_1$ and~$B_1$ for any ground state is given by
\begin{equation}
E_\mathcal{N}^{A_1\mid B_1} \left( \rho_{A_1 B_1} \right)
= 0
\end{equation}
where $\rho_{A_1 B_1} = \tr_{A_2 B_2} \left|\psi\right\rangle\left\langle\psi\right|$.
\end{proposition}

\begin{proof}
We have the decomposition of the generic ground state from proposition\nobreakspace \ref {prop:more than 1 region non-contr:refined partition:decomposition of general GS}.
So, the reduced density matrix over~$A$ and~$B_1$ is
\begin{align}
&\rho_{A_1 B_1}	\nonumber \\
&= \tr_{A_2 B_2} \left|\psi\right\rangle\left\langle\psi\right|	\nonumber \\
&=
\sum_{i=\identity,e,m,em} \left|c_i\right|^2
{\mu_\textnormal{max}}^{-1} {\gamma_\textnormal{max}}^{-1}
{\nu_\textnormal{max}}^{-1} {\tau_\textnormal{max}}^{-1}	\nonumber \\
&\quad
\sum_{\nu=1}^{\nu_\textnormal{max}}
\sum_{\gamma=1}^{\gamma_\textnormal{max}}
\left|e^{\left(\psi_i\right)}_{\nu\gamma}\right\rangle\left\langle e^{\left(\psi_i\right)}_{\nu\gamma}\right|_{A_1}
 \otimes
\sum_{\mu=1}^{\mu_\textnormal{max}}
\sum_{\tau=1}^{\tau_\textnormal{max}}
\left|f^{\left(\psi_i\right)}_{\mu\tau}\right\rangle\left\langle f^{\left(\psi_i\right)}_{\mu\tau}\right|_{B_1}
\end{align}
This state is explicitly separable, that is unentangled, and therefore
\begin{equation}
E_\mathcal{N}^{A_1\mid B_1} \left( \rho_{A_1 B_1} \right)
= 0
\end{equation}
as required.
\end{proof}

\section{Discussion}

\begin{figure}[t]
\includegraphics[width=8.6cm]{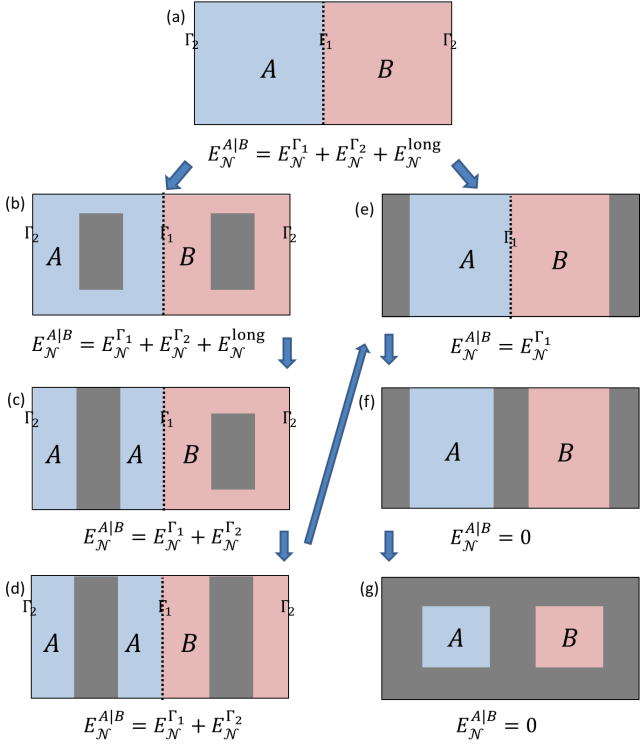}
\caption{
Sequence of partitions of the torus into regions $A$ and $B$ and (possibly, in grey) a complementary region $C$ that is traced out. The negativity allows us to explicitly decompose the entanglement between $A$ and $B$ in setting (a) into boundary contributions (which originate near the boundary and are thus insensitive to the tracing out of degrees of freedom away from those boundaries) and long-range contributions. Only the entanglement of $\rho_{AB}$ in settings (a) and (b) receives long-range contributions. The mixed state $\rho_{AB}$ for settings (f) and (g) is not entangled.
}
\label{fig:MoreRegions}
\end{figure}

\subsection{Boundary and long-range entanglement}

Derivations analogous to the ones presented in the two preceding sections lead to an analytical expression for the negativity for a large variety of settings. Fig. \ref{fig:MoreRegions} displays two sequences of such settings. [Notice that several settings in Fig. \ref{fig:MoreRegions} are equivalent to the ones we already considered above].
To ease the notation, in this last section we divide a torus into three (changing) regions $A$, $B$, and $C$, and trace out region $C$. Then we study the negativity $E_{\mathcal{N}}^{A|B}$ of $\rho_{AB}$. As in the previous section, the torus is in an arbitrary ground state $\ket{\psi} = \sum_{i} c_i \ket{\psi_i}$, where the index $i$ labels the possible anyonic fluxes, $i=I,e,m,em$, inside the torus in the horizontal direction,.

The general expression for the negativity for the settings considered in Fig. \ref{fig:MoreRegions} confirms that the entanglement between $A$ and $B$ is made of two types of contributions: entanglement directly associated to each of the boundaries $\Gamma_m$ ($m=1,2$) between regions $A$ and $B$, and long-range entanglement,
\begin{equation}\label{eq:discussion1}
    E_{\mathcal{N}}^{A|B} = \sum_{m} E_{\mathcal{N}}^{\Gamma_m} + E_{\mathcal{N}}^{\mbox{\tiny long}}.
\end{equation}
Here, each boundary $\Gamma_m$ between $A$ and $B$ contributes an amount $E_{\mathcal{N}}^{\Gamma_m}$ that is independent of the particular choice of ground state $\ket{\Psi}$ and that further breaks into the sum of a term proportional to the size $|\Gamma_m|$ of the boundary, and a universal, topological correction $E^{\gamma}_{\mathcal{N}} \equiv - \gamma = - 1$ \cite{Hamma}, where $\gamma \equiv \log_2 \mathcal{D}$ and $\mathcal{D}=2$ is the total quantum dimension of the toric code \cite{TEE1, TEE2}. In turn, the long-range contribution $E_{\mathcal{N}}^{\mbox{\tiny long}}$ depends on the choice of ground state,
\begin{equation}\label{eq:discussion3}
    E_{\mathcal{N}}^{\mbox{\tiny long}} = 2 \log_2 \left( \sum_{i} |c_i| \right) = S_{1/2}(\{|c_i|^2\}).
\end{equation}

We see in Fig. \ref{fig:MoreRegions} (a),(b),(c),(d),(e) that each boundary contribution $E_{\mathcal{N}}^{\Gamma_m}$ is robust to tracing out degrees of freedom away from that boundary (which is why we refer to them as boundary contributions in the first place). Instead, the long-range contribution $E_{\mathcal{N}}^{\mbox{\tiny long}}$ disappears as soon the traced-out region $C$ closes a nontrivial loop in the vertical direction. In this case the different fluxes $i=1,e,m,em$ can be measured in $C$, causing loss of flux coherence, so that as far as anyon fluxes are concerned, $\rho_{AB}$ only contains, at most, classical correlations. Finally, Fig. \ref{fig:MoreRegions} (f), (g) show that if $A$ and $B$ do not share any boundary, then the negativity of $\rho_{AB}$ vanishes.
In particular, one can further show that $\rho_{AB}$ in Fig. \ref{fig:MoreRegions}(f) is in a classically correlated mixed state,
\begin{equation}\label{eq:classical}
   \rho_{AB} = \sum_{i} |c_i|^2 \rho_A^{(i)} \otimes \rho_B^{(i)},
\end{equation}
with
\begin{eqnarray}\label{eq:classical2}
    \tr(\rho_A^{(i)}\rho_A^{(j)}) &=& \delta_{ij}\tr\left((\rho_A^{(i)})^2\right),\\
    \tr(\rho_B^{(i)}\rho_B^{(j)}) &=& \delta_{ij}\tr\left((\rho_B^{(i)})^2\right),
\end{eqnarray}
because one can still measure the flux inside the torus by means of vertical Wilson loops in $A$ and vertical Wilson loops in $B$, and the outcome of such measurements must coincide.
On the other hand, $\rho_{AB}$ in Fig. \ref{fig:MoreRegions}(f) is in a product (i.e. uncorrelated) state independent of the choice of ground state $\ket{\Psi}$,
\begin{equation}\label{eq:classical3}
   \rho_{AB} = \rho_A \otimes \rho_B.
\end{equation}

\subsection{Beyond the toric code}

Let us now consider a (translation invariant) lattice model for a generic topologically ordered phase. This phase will be characterized by some emergent anyonic model, consisting of $N$ anyon types, $i=1,\cdots,N$, corresponding quantum dimensions $d_i$, and total quantum dimension $\mathcal{D} = \sqrt{\sum_i d_i^2}$. For concreteness, let us consider the bipartition of a torus into two non-contractible parts $A$ and $B$ as in Fig. \ref{fig:MoreRegions}(a). We consider the system to be in a generic ground state $\ket{\psi} = \sum_{i=1}^N c_i \ket{\psi_i}$, where $\ket{\psi_i}$ is the ground state with well-defined flux $i$ propagating inside the torus in the horizontal direction. We assume that the width and the length of each site is much larger than the correlation length. The reduced density matrix for region $A$ is then
\begin{equation}\label{eq:Ageneral}
    \rho_A = \bigoplus_{i=1}^N |c_i|^2 \rho_A^{(i)},
\end{equation}
where $\rho_A^{(i)}$ is the reduced density matrix for region $A$ when the ground state is $\ket{\psi_i}$. In this case, the von Neumann entropy reads \cite{additivity}
\begin{equation}\label{eq:Ageneral2}
    S(\rho_A) = \sum_i |c_i|^2 S(\rho_A^{(i)}) - \sum_{i} |c_i|^2 \log_2 (|c_i|^2),
\end{equation}
where we expect
\begin{equation}\label{eq:Ageneral3}
    S(\rho_A^{(i)}) = (a|\Gamma_1| - \gamma_i) + (a|\Gamma_2| - \gamma_i).
\end{equation}
Here $a$ is a non-universal constant, $|\Gamma_m|$ is the size of the boundary $\Gamma_m$, and $\gamma_i = \log_2 (\mathcal{D}/d_i)$ \cite{TEE1}. Putting the last two equations together we arrive at
\begin{equation}\label{eq:Ageneral4}
    S(\rho_A) = \sum_{m=1}^2\left(a|\Gamma_m| - \bar{\gamma} \right) + S(\{|c_i|^2\}),
\end{equation}
where $\bar{\gamma}\equiv \sum_i |c_i|^2 \gamma_i$. This expression is analogous to Eq. \ref{eq:discussion1} for the logarithmic negativity, and differentiates boundary contributions from long-range contributions to the entanglement between regions $A$ and $B$ in Fig. \ref{fig:MoreRegions}(a). If we now fix the flux inside the torus, then for the Renyi entropy $S_{1/2}$, and thus the logarithmic negativity, we expect
\begin{equation}\label{eq:AgeneralNeg}
    E_{\mathcal{N}}^{A|B} = S_{1/2}(\rho_{A}^{(i)}) = \sum_{m} (a' |\Gamma_m| - \gamma_i),
\end{equation}
where $a'$ is another non-universal constant. The Renyi entropy should again remain unchanged when we trace out degrees of freedom in the bulk of regions $A$ and $B$ (at a distance from the boundaries $\Gamma_1$ and $\Gamma_2$ much larger than the correlation length), showing that the entanglement between $A$ and $B$ originates near the two boundaries between $A$ and $B$.

However, we are not able to generalize Eq. \ref{eq:discussion1} to apply for a \textit{generic} anyon model \cite{additivity}. This is still possible, nevertheless, for those cases where the spectrum of $\rho_A^{(i)}$ is independent of the anyonic flux $i$ (which, in particular, requires that the anyon model be Abelian, with all quantum dimensions $d_i=1$, which implies $\gamma_i = \gamma$). Indeed, then we have that additivity of the Renyi entropies implies, for a generic ground state $\ket{\psi}$, that \cite{additivity}
\begin{equation}\label{eq:RenyiHalf}
   E_{\mathcal{N}}^{A|B} = S_{1/2} (\rho_A) = S_{1/2}(\rho_A^{(1)}) + S_{1/2}(|c_i|^2),
\end{equation}
where $S_{1/2} (\rho_A^{(1)}) = \sum_{m} (a' |\Gamma_m| - \gamma)$ are boundary contributions expected to be robust against tracing out of bulk degrees of freedom in the interior of $A$ and $B$, whereas the long-range term $S_{1/2}(|c_i|^2)$ is expected to disappear when a non-contractible region inside $A$ or $B$ is traced out.

\section{Conclusions}

In conclusion, in this paper we have presented analytical calculations of the entanglement negativity $E_{\mathcal{N}}$ in the ground state of the toric code model for a variety of settings, and have used them to explore the structure of entanglement in a topologically ordered system. We have seen that the entanglement of a region $A$ and the rest of the system $B$ is made of boundary contributions that entangle degrees of freedom near the boundary between $A$ and $B$; and (possibly) of a long-range contribution. The later appears when through non-contractible regions $A$ and $B$ there is a flux corresponding to a linear combination of different anyon types.

\begin{acknowledgments}
The authors thank Lukasz Cincio for numerically corroborating the validity of the analytical expressions for some of the calculations presented in this paper, and Oliver Buerschaper and Alioscia Hamma for discussions.

After this work was completed, private communication with C. Castelnovo revealed that he had independently arrived at similar results, see C. Castelnovo, arXiv:1306.4990.

\end{acknowledgments}



\appendix
\section{Extension of the logarithmic negativity}
\label{sec:appd:ext of log neg}

Consider the quantity
\begin{equation}
\mathcal{N}_\alpha^+ \left(\rho\right)
= \frac{1}{2\left(1-\alpha\right)}
\log_2 \tr \left|\rho^{\top_A}\right|^{2\alpha}
\end{equation}
where $\rho$ is any density matrix and $\alpha\in\mathbb{R}$.
For $\alpha = 1/2$, this is exactly the logarithmic negativity $\mathcal{N}_{1/2}^+ = E_\mathcal{N}$.
Interestingly, for any pure state $\rho$, we have the relation that
\begin{equation}
S_\alpha \left(\rho_A\right)
= \mathcal{N}_\alpha^+\left(\rho\right)
\end{equation}
where $S_\alpha \left(\rho_A\right)$ denotes the R\'{e}nyi entropy of order~$\alpha$ of the reduced density matrix over $A$.
Thus, in the case of pure states, we can recover all R\'{e}nyi entropies by considering these negativity-like quantities.

Similarly, the quantity
\begin{equation}
\mathcal{N}_\alpha^- \left(\rho\right)
= \frac{1}{\left(1-2\alpha\right)}
\log_2 \tr\left[ \mathrm{sign}\left(\rho^{\top_A}\right) \left|\rho^{\top_A}\right|^{2\alpha} \right]
\end{equation}
can be shown to satisfy
\begin{equation}
S_{2\alpha} \left(\rho_A\right)
= \mathcal{N}_\alpha^-\left(\rho\right)
\end{equation}
in the case where $\rho$ is a pure state.
However, it is still an open question as to whether these quantity can be extended to mixed states such that it is an entanglement monotone.

\section{Some proofs}

\begin{proof}[Proof for proposition\nobreakspace \ref {prop:one region contr:4-partite decomposition:decomposition of GS1}]
Following the arguments outlined in the proof of proposition\nobreakspace \ref {prop:one region contr:schmidt decomposition:decomposition of GS1},
we first consider the ground state~$\left|\psi_0\right\rangle$ in Eq. \ref{eq:psi0}.

\begin{widetext}
The ground state~$\left|\psi_0\right\rangle$ can be written as
\begin{align}
\left|\psi_0\right\rangle
=&\; N
\left(\sum\textnormal{all loops of $\sigma^z$ operators in $A_1\cup A_2\cup B_1\cup B_2$}\right)
\left|+\cdots+\right\rangle	\nonumber \\
=&\; N
\left(\sum_\textnormal{all} \textnormal{loops of $\sigma^z$ operators crossing the boundary $A_1\mid B_1$ up to equivalance class}\right)_{A_1 B_1}	\nonumber \\
&\cdot \left[
\left(\sum_\textnormal{all} \textnormal{loops crossing $A_1\mid A_2$ up to equiv}\right)_{A_1 A_2}
\cdot
\left(\sum_\textnormal{all} \textnormal{loops in $A_1$}\right)_{A_1} \left|+\right\rangle_{A_1}
\otimes
\left(\sum_\textnormal{all} \textnormal{loops in $A_2$}\right)_{A_2} \left|+\right\rangle_{A_2}
\right]	\nonumber \\
&\cdot \left[
\left(\sum_\textnormal{all} \textnormal{loops crossing $B_1\mid B_2$ up to equiv}\right)_{B_1 B_2}
\cdot
\left(\sum_\textnormal{all} \textnormal{loops in $B_1$}\right)_{B_1} \left|+\right\rangle_{B_1}
\otimes
\left(\sum_\textnormal{all} \textnormal{loops in $B_2$}\right)_{B_2} \left|+\right\rangle_{B_2}
\right]
\end{align}
where the equivalence class is such that 2 loops are equivalent if one can be deformed into the other by loop operators that individually act upon one region only.

As before, for the loops crossing some boundary, we can consider only loops made up of plaquette operators on that boundary.
\begin{align}
&\left|\psi_0\right\rangle	\nonumber \\
&=\; N
\left(\sum^{\mu_\textnormal{max}}_{\mu=1} \textnormal{prod of plaquette ops on $A_1\mid B_1$, }P_\mu\right)_{AB}	 \nonumber \\
&\quad\cdot \left[
\left(\sum^{\nu_\textnormal{max}}_{\nu=1} \textnormal{prod of plaquette ops on $A_1\mid A_2$, }Q_\mu\right)_{A_1 A_2}
\cdot
\left(\sum_\textnormal{all} \textnormal{loops in $A_1$}\right)_{A_1} \left|+\right\rangle_{A_1}
\otimes
\left(\sum_\textnormal{all} \textnormal{loops in $A_2$}\right)_{A_2} \left|+\right\rangle_{A_2}
\right]	\nonumber \\
&\quad\cdot \left[
\left(\sum^{\tau_\textnormal{max}}_{\tau=1} \textnormal{prod of plaquette ops on $B_1\mid B_2$, }R_\tau\right)_{B_1 B_2}
\cdot
\left(\sum_\textnormal{all} \textnormal{loops in $B_1$}\right)_{B_1} \left|+\right\rangle_{B_1}
\otimes
\left(\sum_\textnormal{all} \textnormal{loops in $B_2$}\right)_{B_2} \left|+\right\rangle_{B_2}
\right]
\end{align}
\end{widetext}

Breaking up the operators $P$, $Q$ and $R$ into two parts where each individually act only on one partition, and rewriting the resulting state as $\left|e_{\mu\nu}\right\rangle_{A_1}$, $\left|f^{\left(\psi\right)}_{\mu\tau}\right\rangle_{B_1}$, $\left|g_\nu\right\rangle_{A_2}$, $\left|h_\tau\right\rangle_{B_2}$ gives the desired result for the ground state~$\left|\psi_0\right\rangle$ up to the normalisation,
\begin{equation}
\left|\psi_0\right\rangle
\propto \sum_\mu
\left( \sum_\nu \left|e_{\mu\nu}\right\rangle_{A_1} \left|g_\nu\right\rangle_{A_2} \right)
\left( \sum_\tau \left|f^1_{\mu\tau}\right\rangle_{B_1} \left|h_\tau\right\rangle_{B_2} \right)
.
\end{equation}
Normalising the state such that $\left\langle\psi_0\mid\psi_0\right\rangle = 1$ fixes the constants and completes the proof.
\end{proof}

\begin{proof}[Proof for proposition\nobreakspace \ref {prop:more than 1 region non-contr:bipartition:decomposition of GSident}]
First consider the ground state~$\left|\psi_{\identity}\right\rangle$.

\begin{widetext}
The ground state~$\left|\psi_{\identity}\right\rangle$ can be written as
\begin{align}
\left|\psi_{\identity}\right\rangle
=&\; N
\left(\sum\textnormal{all loops in $A\cup B$ that are a composition of contractible loops and the vertical loop}\right)
\left|+\cdots+\right\rangle	\nonumber \\
=&\; N
\left(\sum\textnormal{all loops as above crossing the boundary $A\mid B$ up to equivalance class}\right)	\nonumber \\
&\cdot
\left(\sum\textnormal{all loops of $\sigma^z$ operators in $A$}\right)_A \left|+\cdots+\right\rangle_A
\otimes
\left(\sum\textnormal{all loops of $\sigma^z$ operators in $B$}\right)_B \left|+\cdots+\right\rangle_B
\end{align}
where the equivalence class is such that 2 loops are equivalent if one can be deformed into the other by loop operators in $A$~and/or~$B$ only.
Note that the superposition of loop operators in $\left|\psi_{\identity}\right\rangle$ does not include the horizontal non-contractible loop.
Also note that `all loops of $\sigma^z$ operators in $A$' (and $B$) still refer to both contractible and non-contractible loops in $A$ (as well as $B$).

Thus we can consider only loops made up of plaquette operators on the boundary.
\begin{align}
\left|\psi_{\identity}\right\rangle
&=\; N
\left(\sum^{\mu_\textnormal{max}}_{\mu=1} \textnormal{prod of plaquette ops on boundary~$1$, }P_\mu\right)
\left(\sum^{\gamma_\textnormal{max}}_{\gamma=1} \textnormal{prod of plaquette ops on boundary~$2$, }Q_\gamma\right)	 \nonumber \\
&\quad\quad \cdot
\left(\sum\textnormal{all loops in A}\right)_A \left|+\cdots+\right\rangle_A
\otimes
\left(\sum\textnormal{all loops in B}\right)_B \left|+\cdots+\right\rangle_B
\end{align}
\end{widetext}

Breaking up the operators $P$ and $Q$ into two parts where each individually act only on one partition, and
rewriting the resulting state as $\left|e_{\mu\gamma}\right\rangle_A$, $\left|f_{\mu\gamma}\right\rangle_B$ gives the desired result for the ground state~$\left|\psi_{\identity}\right\rangle$ up to the normalisation $N$.
Normalising the state such that $\left\langle\psi_{\identity}\mid\psi_{\identity}\right\rangle = 1$ fixes the constants and completes the proof.
\end{proof}

\begin{proof}[Proof of proposition\nobreakspace \ref {prop:more than 1 region non-contr:refined partition:negativity:ent between A and B1}]
We have the decomposition of the generic ground state from proposition\nobreakspace \ref {prop:more than 1 region non-contr:refined partition:decomposition of general GS}.
So, the reduced density matrix over~$A$ and~$B_1$ is
\begin{align}
&\rho_{A B_1}	\nonumber \\
&= \tr_{B_2} \left|\psi\right\rangle\left\langle\psi\right|	\nonumber \\
&=
\sum_{i=\identity,e,m,em} \left|c_i\right|^2
{\mu_\textnormal{max}}^{-1} {\gamma_\textnormal{max}}^{-1}
{\nu_\textnormal{max}}^{-1} {\tau_\textnormal{max}}^{-1}	\nonumber \\
&\quad\quad
\sum_{\mu,\mu'=1}^{\tau_\textnormal{max}}
\sum_{\gamma=1}^{\tau_\textnormal{max}}
\sum_{\nu,\nu'=1}^{\tau_\textnormal{max}}
\sum_{\tau=1}^{\tau_\textnormal{max}}
\left|e^{\left(\psi_i\right)}_{\nu\gamma}\right\rangle\left\langle e^{\left(\psi_i\right)}_{\nu'\gamma}\right|_{A_1}	 \nonumber \\
&\quad\quad\quad\quad\quad\quad\quad
\otimes
\left|f^{\left(\psi_i\right)}_{\mu\tau}\right\rangle\left\langle f^{\left(\psi_i\right)}_{\mu'\tau}\right|_{B_1}
\otimes
\left|g^{\left(\psi_i\right)}_{\mu\nu}\right\rangle\left\langle g^{\left(\psi_i\right)}_{\mu'\nu'}\right|_{A_2}
\end{align}
Taking the partial transpose and squaring, we have that
\begin{align}
&\left( {\rho_{A B_1}}^{T_A} \right)^2	\nonumber \\
&=
\sum_{i=\identity,e,m,em} \left|c_i\right|^4
{\mu_\textnormal{max}}^{-2} {\gamma_\textnormal{max}}^{-2}
{\nu_\textnormal{max}}^{-2} {\tau_\textnormal{max}}^{-2}	\nonumber \\
&\quad\quad
\sum_{\mu,\mu'=1}^{\tau_\textnormal{max}}
\sum_{\gamma=1}^{\tau_\textnormal{max}}
\sum_{\nu,\nu'=1}^{\tau_\textnormal{max}}
\sum_{\tau=1}^{\tau_\textnormal{max}}
\left|e^{\left(\psi_i\right)}_{\nu'\gamma}\right\rangle\left\langle e^{\left(\psi_i\right)}_{\nu\gamma}\right|_{A_1}	 \nonumber \\
&\quad\quad\quad\quad\quad\quad\quad
\otimes
\left|f^{\left(\psi_i\right)}_{\mu\tau}\right\rangle\left\langle f^{\left(\psi_i\right)}_{\mu\tau}\right|_{B_1}
\otimes
\left|g^{\left(\psi_i\right)}_{\mu'\nu'}\right\rangle\left\langle g^{\left(\psi_i\right)}_{\mu'\nu}\right|_{A_2}
\end{align}
and therefore we have that
\begin{align}
&\left| {\rho_{A B_1}}^{T_A} \right|	\nonumber \\
&= \sqrt{ {{\rho_{A B_1}}^{T_A}}^\dagger {\rho_{A B_1}}^{T_A} }
= \sqrt{ \left( {\rho_{A B_1}}^{T_A} \right)^2 }	\nonumber \\
&=
\sum_{i=\identity,e,m,em} \left|c_i\right|^2
{\mu_\textnormal{max}}^{-1} {\gamma_\textnormal{max}}^{-1}
{\nu_\textnormal{max}}^{-1} {\tau_\textnormal{max}}^{-1}	\nonumber \\
&\quad\quad
\sum_{\mu,\mu'=1}^{\tau_\textnormal{max}}
\sum_{\gamma=1}^{\tau_\textnormal{max}}
\sum_{\nu,\nu'=1}^{\tau_\textnormal{max}}
\sum_{\tau=1}^{\tau_\textnormal{max}}
\left|e^{\left(\psi_i\right)}_{\nu'\gamma}\right\rangle\left\langle e^{\left(\psi_i\right)}_{\nu'\gamma}\right|_{A_1}	 \nonumber \\
&\quad\quad\quad\quad\quad\quad\quad
\otimes
\left|f^{\left(\psi_i\right)}_{\mu\tau}\right\rangle\left\langle f^{\left(\psi_i\right)}_{\mu\tau}\right|_{B_1}
\otimes
\left|g^{\left(\psi_i\right)}_{\mu'\nu'}\right\rangle\left\langle g^{\left(\psi_i\right)}_{\mu'\nu}\right|_{A_2}
\end{align}
where the last equality follows because the state is already in its eigenvalue decomposition.
Therefore we have the negativity as
\begin{align}
&E_\mathcal{N}^{A\mid B_1} \left( \rho_{A B_1} \right)	\nonumber \\
&= \log_2\tr \left| {\rho_{A B_1}}^{T_A} \right|	\nonumber \\
&= \log_2 \left(
\sum_{i=\identity,e,m,em} \left|c_i\right|^2
{\mu_\textnormal{max}}
\right)
	\nonumber \\
&= \log_2 {\mu_\textnormal{max}}
= n_{A_2 B_1} - 1
\end{align}
as required.
\end{proof}


\begin{thebibliography}{74}
\bibitem{additivity}
Notice that the Von Neumann entropy is the only Renyi entropy that fullfils the chain rule of conditional probability, $S(\{p_i r^{(i)}_{j}\}) = \sum_i p_i S(\{r^{(i)}_j\}) + S(\{p_i\})$, where $p_i$ and $r_j^{(i)}$ (for each value of $i$) are probability distributions. We have used this relation to arrive to Eq. \ref{eq:Ageneral2}. For any Renyi entropy of index $q\neq 1$, when the probability distributions $r^{(i)}_{j}$ are independent of $i$ we obtain additivity of Renyi entropies to write $S_q(\{p_i r^{(i)}_j\}) = S_q (\{p_i\}\times \{r^{(1)}_j\}) = S_q(\{r^{(1)}_j\}) + S_q(\{p_i\})$. We have used this relation to arrive to Eq. \ref{eq:RenyiHalf}.
\end{thebibliography}
\end{document}